\documentclass[a4paper]{article}

\usepackage[T1]{fontenc}
\usepackage[latin1]{inputenc}
\usepackage{amsthm,amsmath,mathrsfs,amsfonts}
\usepackage{graphics}
\usepackage[all]{xy}
\usepackage{fullpage}
\usepackage{hyperref}

\hypersetup{pdfstartview=FitH,%
   pdfauthor={B.~Michel},%
   pdftitle={Geometric invariance of mass-like asymptotic invariance},%
   pdfkeywords={Asymptotic invariants, total charges, mass},%
}

\newcommand{\resply}{respectively}

\newcommand{\langue}[1]{\textit{#1}}

\newcommand{\abs}[1]{\left\lvert #1\right\rvert}
\newcommand{\norm}[1]{\left\lVert #1\right\rVert}
\newcommand{\set}[1]{\left\{ #1\right\}}

\newcommand{\itbul}{\item[$\bullet$]}

\newcommand{\veps}{\varepsilon}
\newcommand{\ud}{\mathrm{d}}
\newcommand{\uD}{\mathrm{D}}
\newcommand{\N}{\mathbb{N}}

\newcommand{\R}{\mathbb{R}}

\DeclareMathOperator{\Id}{Id}

\DeclareMathOperator{\tr}{tr}
\DeclareMathOperator{\Diff}{Diff}
\newcommand{\Lieder}{\mathscr{L}}
\newcommand{\DP}[2]{\frac{\partial #1}{\partial #2}}
\newcommand{\Sph}{\mathbb{S}}
\newcommand{\Hyp}{\mathbb{H}}

\newcommand{\defeq}{:=}
\newcommand{\eqdef}{=:}
\newcommand{\nablao}{\nabla}

\newcommand{\G}{\Gamma}
\newcommand{\Scal}{\mathrm{Scal}}
\newcommand{\Rm}{\mathrm{R}}
\newcommand{\Ric}{\mathrm{Ric}}
\newcommand{\Vol}{\mathrm{Vol}}
\newcommand{\vol}{\mathrm{vol}}
\DeclareMathOperator{\divg}{div}

\newcommand{\be}{\begin{equation}}
\newcommand{\ee}{\end{equation}}
\newcommand{\bi}{\begin{itemize}}
\newcommand{\ei}{\end{itemize}}

\newtheorem*{Canclemma}{Cancellation Lemma}
\newtheorem{engthm}{Theorem}[section]
\newtheorem{prop}[engthm]{Proposition}
\newtheorem{engcor}[engthm]{Corollary}

\theoremstyle{definition}
\newtheorem{engdef}[engthm]{Definition}

\theoremstyle{remark}
\newtheorem{engrk}[engthm]{Remark}

\author{B.~Michel\thanks{I3M -- Universit\'e Montpellier 2 -- France.
E-mail: \texttt{benoit.michel@math.univ-montp2.fr}}}
\title{Geometric invariance of mass-like asymptotic invariants}
\date{}

\begin{document}

\maketitle

\begin{abstract}
We study coordinate-invariance of some asymptotic invariants such as the ADM mass
or the Chru\'sciel-Herzlich momentum, given by an integral over a ``boundary at infinity''.
When changing the coordinates at infinity, some terms in the change of integrand
do not decay fast enough to have a vanishing integral at infinity; but they may be gathered
in a divergence, thus having vanishing integral over any closed hypersurface.
This fact could only be checked after direct calculation
(and was called a ``curious cancellation''). We give a conceptual explanation thereof.
\end{abstract}

\section{Introduction}

General Relativity has introduced a new kind
of geometric invariants that depend on the geometry ``at infinity'' of
a non compact Riemannian manifold. Given such a $(M,g)$, and a reference
Riemannian metric $g_0$ to which $g$ is asymptotic, these invariants
are formally defined as an integral over the ``boundary at infinity'' of $M$
of a field of $1$-forms $\mathbb{U}(g,g_0)$:
\begin{equation}
m(g,g_0)\defeq\oint_{S_\infty}\mathbb{U}(g,g_0)(\nu)\ud S. \label{defmass0}
\end{equation}
The integral over $S_\infty$
is understood as a limit
\[
\lim_{r\to\infty}\oint_{S_r}\mathbb{U}(g,g_0)(\nu)\ud S
\]
where $(S_r)_r$ is a family of closed hypersurfaces enclosing the whole of $M$ when
$r\to\infty$,
and $\nu$ and $\ud S$ respectively are the outer unit normal and induced volume
measure of $S_r$ with respect to $g_0$. In fact, $g$ and $g_0$ could encode
other geometric data, for example first and second fundamental forms of
a space-like hypersurface in a space-time. The ADM and Abbott-Deser
energy-momentum \cite{ADM-mass}, \cite{AbbottDeser-mass}, mathematically
studied in terms of Cauchy data in \cite{Chrusciel-mass}, \cite{Bartnik-mass},
\cite{CH-massAH}, \cite{ChrNagy-massAAdS}, belong to that category. We shall
call such invariants \emph{total charges}.

In general, many mutually isometric but distinct $g_0$'s exist and are asymptotic
to the given $g$, for instance pulled-back $\Psi^*g_0$ when $\Psi$ is any
diffeomorphism of $M$ suitably asymptotic to the identity.
However, in the examples mentioned above, with appropriate decay conditions,
it is proven (\cite{Chrusciel-mass}, \cite{Bartnik-mass}, \cite{CH-massAH},
\cite{ChrNagy-massAAdS}) that the resulting $m(g,g_0)$ does not depend on the
particular chosen $g_0$. The proof relies on an algebraic fact that we recall
below, in the simply stated case of the ADM mass of an asymptotically flat manifold.

The Riemannian manifold $(M,g)$ is said to be asymptotically flat when there
exists a set of coordinates $(x^i)_i$ defined outside a compact subset of $M$,
in which the metric coefficients satisfy $g_{ij}=\delta_{ij}+e_{ij}$,
where $e_{ij}=O(\abs{x}^{-\tau})$ and similar decay holds for $\partial_ke_{ij}$.
The metric $g$ is then asymptotic to the flat metric $g_0\defeq\sum d x^{i\,2}$.
One defines the $1$-form
\begin{align*}
\mathbb{U}_i(g,g_0)\!&\;=\sum\nolimits_j\partial_jg_{ij}-\partial_ig_{jj}\\
\text{\emph{i.e.} }\mathbb{U}(g,g_0)&\defeq\divg_{g_0}g-d(\tr_{g_0}g)
\end{align*}
and the ADM mass of $g$, \emph{a priori} with respect to this particular
chart at infinity, is given by formula \eqref{defmass0}.

Let $\hat{x}^{i}$ be other coordinates in which $g$ is asymptotically flat.
It is an intuitive---but non trivial---theorem that, maybe after rotating and
translating the $\hat{x}^i$'s, one has $\hat{x}^i-x^i\eqdef v^i=O(\abs{x}^{1-\tau})$, and
similar decay holds for two derivatives of $v^i$. An easy calculation then shows that
the $1$-forms $\mathbb{U}$ defining the ADM mass in the two coordinate sets respectively
are related by
\[
\mathbb{U}(g,g_0)-\mathbb{U}(g,\hat{g}_0)=\sum\nolimits_{i,j}
   \partial_i\bigl(\partial_jv^i-\partial_iv^j\bigr)d x^j+O(\abs{x}^{-2\tau-1}).
\]
Provided $\tau>\frac{n-2}{2}$, the decay of the last term is faster than the critical rate
$r^{1-n}$ (the volume of large coordinates spheres), so this term does not contribute to
the limit of the integrals over large spheres.
The first term in the right-hand side, however, decays slower than
the critical rate. Invariance of the ADM mass comes from the fact that it happens to be the
divergence of an alternating $2$-form, so that its integral over any closed hypersurface
vanishes. But this divergence curiously only appears in a direct calculation, knowing the explicit
formula for $\mathbb{U}(g,g_0)-\mathbb{U}(g,\hat{g}_0)$, so that it was for example called a
``curious cancellation'' by R. Bartnik \cite{Bartnik-mass}.

In the asymptotically hyperbolic setting, the definition
of the Chru\'sciel-Herzlich mass \cite{CH-massAH} is more elaborate and uses an auxiliary
function in the integrand $\mathbb{U}$. Yet its invariance also relies on the appearance of
a divergence gathering the terms
above the critical decay rate, and still only in a direct calculation.

We give in this note a conceptual explanation to that ``curious cancellation''.
Simplifying a little, we follow the three following steps:
\begin{enumerate}
\item \label{step.defcharge} We give a general construction of a charge integrand $\mathbb{U}$
and the related total charges,
\item \label{step.diffintgd} We obtain a general expression for the difference of charge
integrands computed in two different charts at infinity, modulo fast-decaying terms,
\item \label{step.curcanc} We prove that this expression is the divergence of a field
of alternating $2$-forms.
\end{enumerate}
Point \ref{step.defcharge} follows M.~Herzlich \cite[§3.2]{Herzlich-massAH}---see also
\cite{CJL-Bondimass} and \cite{Maerten-PosMomt}. In fact Points \ref{step.diffintgd}
and \ref{step.curcanc} answer Question $3.5$ in the first reference. Both appear to be
very simple. It must be noticed however that in Point \ref{step.diffintgd} we pick up
an expression among many other possible ones, all equal modulo fast decaying terms. To
find the one appropriate for Point \ref{step.curcanc}, one needs to guess in advance
the phenomenon that occurs there, despite it may not be properly formulated without
knowing the accurate formula of Point \ref{step.diffintgd}. The interest of the construction
of Point \ref{step.defcharge} is then justified \langue{a posteriori} only.

Our study is unrelated to Hamiltonian or
Lagrangian formalism, see for example \cite{WaldZoupas}, which
builds uniquely defined Hamiltonians, but only in restriction
to phase space, \emph{i.e.} geometric data that satisfy some
constraints. (Moreover the dependence on the ``appropriate decay conditions'',
as usually named in the literature, is not very clear at a formal level.) Our construction
is closer in spirit to that of \cite{AndersonTorre}, but different, in that we are not
interested in conservation laws, but only in geometric invariants that are independent,
to a certain extent, of the background. One may check indeed that the formulae given in
\cite{AndersonTorre} are not those obtained following the construction presented
here.\footnote{A general result about gauge invariance seems to be claimed in
\cite{AndersonTorre}, but unfortunately does not appear in the literature, even as a
preprint.}

The outline of the article is the following. In Section \ref{sec.charge}, we give the
construction of the charge integrands $\mathbb{U}$ considered in this paper, see
Definition \ref{defU}, and we use it to define a total charge in Defintion \ref{defcharge}.
In Section \ref{sec.geoinv} we prove the main result, Theorem \ref{mainres}, according
to which the total charge is invariant under a diffeomorphism suitably asymptotic to the
identity at infinity.  It covers Points \ref{step.diffintgd} and \ref{step.curcanc} above,
which correspond respectively to formula \eqref{diffU} and the
\hyperref[lemma.canc]{Cancellation Lemma} page \pageref{lemma.canc}.
The meaning of ``suitably asymptotic to the identity'' is quite technical
to state accurately. Since it is not the main interest of this paper, the precise
discussion is postponed to
the appendix \ref{sec.controlR1}. Before that, we show in Section \ref{sec.ex}
how our reasoning apply to already known invariants.

\section{Total charge} \label{sec.charge}

Let $M$ be a non compact differential manifold without boundary. The geometric data we
consider here are sections $h=(g,k)$ of a bundle $H=\mathscr{M}\times_ME$
over $M$, where $\mathscr{M}$ is the bundle of metrics and $E$ a natural tensor
bundle. The role of local charge density will be played by
a natural tensor-valued differential operator
\[
\Phi : \Gamma(H) \longrightarrow \Gamma(F)
\]
where $F$ is a natural tensor bundle over $M$, and the letter $\Gamma$ denotes the
space of $C^\infty$ sections. This operator needs not be linear, but we
require invariance under diffeomorphisms: \emph{i.e.} for all sections $h$ of
$H$ and all diffeomorphisms $\Psi$ of $M$,
\[
\Psi^*\bigl(\Phi(h)\bigr) = \Phi(\Psi^*h).
\]

For clarity we will suppose that the background data are given on an other manifold
$M_0$. Let us write $\mathscr{M}_0$, $E_0$, and $F_0$ for the bundles over $M_0$
that correspond to $\mathscr{M}$, $E$, and $F$ respectively. 
We denote by $h_0=(g_0,k_0)$ a section of $H_0\defeq\mathscr{M}_0 \times_{M_0} E_0$ that
will be used as reference. Because of naturality, $\Phi$ is defined on $M_0$; we set
$\Phi_0\defeq\Phi(h_0)$. We also introduce
the dot product $\langle\cdot,\cdot\rangle_0$ and norm $\abs{\cdot}_0$
induced by $g_0$ on natural tensor bundles over $M_0$.

We now define $\mathbb{U}$.

\begin{engdef} Let $\uD\Phi_0$ be the linearization of $\Phi$ at $h_0=(g_0,k_0)$
and $\uD\Phi_0^*$ be its formal adjoint with respect to $g_0$.

For a $C^\infty$ section $V$ of $F_0$ and a $C^\infty$ section $\eta$ of
$S^2M_0\times_{M_0}E_0$, the charge integrand $\mathbb{U}(V,\eta)$ is the $1$-form
appearing in the following integration-by-part formula:
\[
\bigl\langle V,\uD\Phi_0(\eta)\bigr\rangle_0 = \divg_0 \mathbb{U}(V,\eta) +
 \bigl\langle\uD\Phi_0^*V,\eta\bigr\rangle_0.
\]
Here $\divg_0$ is the
$g_0$-divergence operator: if $\nabla$ is the Levi-Civita connection of
$g_0$ and $\alpha$ is a form field, $\divg_0 \alpha=\nabla^i\alpha_i$.
\label{defU}\end{engdef}
\begin{engrk} The operator $\mathbb{U}$ is a differential operator, linear
and of order $1$ less than $\Phi$ in each of its arguments.
\label{opU}\end{engrk}

To use $\mathbb{U}$ to define asymptotic invariants, we are interested in the situation
where the outside of a compact subset of $M$ is diffeomorphic to the outside of a compact
subset of $M_0$, \emph{i.e.} when the following definition is not empty:
\begin{engdef}
A \emph{diffeomorphism at infinity} is a diffeomorphism
\[
\Psi : M_0-K_0 \longrightarrow M-K
\]
where $K_0$ and $K$ are compact subsets of respectively $M_0$ and $M$.
\label{defdiffinfty}\end{engdef}

Let $h$ be a section of $H$, and set $e\defeq\Psi^*h-h_0$ outside $K_0$. We will require
$e$ to tend to $0$ at infinity (in a sense to be made precise below). This justifies the
use of the following Taylor formula, which however always makes sense:
$\Phi(\Psi^*h)-\Phi_0=\uD\Phi_0(e)+Q(e)$, where $Q(e)$ is the quadratic and higher order
remainder.
Contracting with a test-section $V$ of $F_0$ to get a numerical value, we obtain:
\begin{align}
\bigl\langle V,\Phi(\Psi^*h)-\Phi_0\bigr\rangle_0
   &= \bigl\langle V,\uD \Phi_0(e)\bigr\rangle_0
     + Q\bigl(V,e\bigr)\notag\\
   &= \divg_{0}\mathbb{U}(V,e) +\big\langle\uD \Phi_0^*(V),e\big\rangle_0
      + Q(V,e) \label{eqU}
\end{align}
where $Q(V,e)\defeq\langle V,Q(e)\rangle_0$ for short, and Definition \ref{defU} has been
used. Let us introduce
\[
\mathscr{N}_0 \defeq \set{V \in\Gamma(F_0) | \uD\Phi_0^*V=0}.
\]
When $V\in\mathscr{N}_0$, the right-hand side of \eqref{eqU} contains only
a divergence plus terms of quadratic and higher order in $e$. Thus we are interested
in the following class of data $h$:
\begin{engdef}
A section $h$ of $H$ is said to have \emph{well-defined total charge} with respect
to a diffeomorphism at infinity $\Psi$ and to $V\in\mathscr{N}_0$ (respectively
to a subspace $\mathscr{N}'_0\subset\mathscr{N}_0$) when:
\begin{enumerate}
\item \label{hyp.totch.lc} $\bigl\langle V,\Phi(\Psi^*h)-\Phi_0\bigr\rangle_0$
is integrable (with respect to the volume density induced by $g_0$),
\item \label{hyp.totch.q} writing $e=\Psi^*h-h_0$, $Q\big(V,e)$ is integrable
\end{enumerate}
(respectively when \ref{hyp.totch.lc} and \ref{hyp.totch.q} hold for all
$V\in\mathscr{N}'_0$).
\label{hyp.totch}\end{engdef}
For such $h$'s one imitates the classical definitions:
\begin{engdef} \label{defcharge}
Let $h$ have well-defined total charge with respect to some diffeomorphism at
infinity $\Psi$ and some $V\in \mathscr{N}_0$.

Let $(B_k)_{k\in\N}$ be an non-decreasing
exhaustion of $M_0$ such that each $B_k$ has smooth compact boundary $S_k$.
One defines the \emph{total charge} as the following limit
\[
m(h,\Psi,V) \defeq
   \lim_{k \to \infty} \oint_{S_k} \mathbb{U}\big(V,\Psi^*h-h_0\big)(\nu)\ud S
\]
where $\nu$ and $\ud S$ are the outer normal and surface measure of $S_k$
with respect to $g_0$.
\end{engdef}
The limit is finite and independent of the chosen exhaustion
$(B_k)$. Let us recall the classical proof of this fact.
\begin{proof}
When $V\in\mathscr{N}_0$, Equation \eqref{eqU} becomes
\[
\divg_0\mathbb{U}(V,e)=\bigl\langle V,\Phi(\Psi^*h)-\Phi_0\bigr\rangle_0
-Q(V,e).
\]
There exists $k_1$ such that for $k\geq k_1$,
$M_0-B_k$ is included in the domain $M_0-K_0$ of $\Psi$.
Integrating over $B_k-B_{k_1}$ with respect to the volume element $\ud\vol_0$
of $g_0$:
\[
\oint_{S_k}\mathbb{U}(V,e)(\nu)\ud S=
   \oint_{S_{k_1}}\mathbb{U}(V,e)(\nu)\ud S
   +\int_{B_k-B_{k_1}}\Bigl[\bigl\langle V,\Phi(\Psi^*h)-\Phi_0\bigr\rangle_0
-Q(V,e)\Bigr]\ud\vol_0.
\]
Definition \ref{hyp.totch} insures that the
right-hand side has finite limit when $k\to\infty$. Moreover the
formula shows that the right-hand side is independent of $k_1$; in fact $B_{k_1}$
and $S_{k_1}$ there could be replaced by any $B$ and $S=\partial B$ respectively,
without changing its value (provided $B$ is large enough for $\Psi$ to be defined
on $M_0-B$). Thus the limit
\[\oint_{S_{k_1}}\mathbb{U}(V,e)(\nu)\ud S +
   \int_{M_0-B_{k_1}}\Bigl[\bigl\langle V,\Phi(\Psi^*h)-\Phi_0\bigr\rangle_0
-Q(V,e)\Bigr]\ud\vol_0
\]
is independent of the exhaustion $(B_k)$.
\end{proof}
\begin{engrk} \label{rk.equiv}
Let $G_0$ be the group of diffeomorphisms of $M_0$ fixing $h_0$. It
acts on $\mathscr{N}_0$ by pull-back. The total charge is
of particular interest when $h$ satisfies Defintion \ref{hyp.totch} with respect to some
$\Psi$ and to a subspace $\mathscr{N}_0'\subset\mathscr{N}_0$ invariant under $G_0$.
Indeed one straightforwardly checks (due to the independence with respect to the exhaustion
$B_k$) that the total charge functional $m$ is then a $G_0$-equivariant linear form on
$\mathscr{N}_0'$, in the sense that for all $A\in G_0$ and $V\in\mathscr{N}_0'$
\[
m(h,\Psi,A^*V)=m(h,\Psi\circ A^{-1},V).
\]
\end{engrk}
\begin{engrk}
In general, $\mathscr{N}_0$ may contain only the zero section of $F_0$, making
Definitions \ref{hyp.totch} and \ref{defcharge} trivial. There are however
many examples were it does not, see Section \ref{sec.ex}.
\end{engrk}

\section{Geometric invariance} \label{sec.geoinv}

We compare the total charges given by two diffeomorphisms at infinity
\[
\xymatrix{ & (M,h) \\
(M_0,h_0) \ar[ur]^{\Psi_1} & & (M_0,h_0) \ar[ul]_{\Psi_2}}
\]
such that $h$ has well-defined total charge with respect to some $V\in
\mathscr{N}_0$ and to both $\Psi_1$ and $\Psi_2$. We assume that $\Psi\defeq\Psi_1^{-1}
\circ\Psi_2$ tends to the identity at infinity, in the following sense.

Let $\exp:TM_0\to M_0$ be the exponential map of $g_0$. For a section $\zeta$ of $TM_0$, let us
write $\exp\circ\zeta : x \mapsto \exp_x\bigl(\zeta(x)\bigr)$. Since $\exp\circ0=\Id_{M_0}$,
there is a $C^1$ neighborhood $\mathscr{U}$ of the zero section of $TM_0$ such that for
all $\zeta\in\mathscr{U}$, $\exp\circ\zeta$ is a diffeomorphism such that
$1/4g_0\leq(\exp\circ\zeta)^*g_0\leq 4g_0$ (see Proposition \ref{prop.R1}-part \ref{propR1.diff}).

\begin{engdef}
A diffeomorphism at infinity $\Psi:M_0-K_1\longrightarrow M_0-K_2$ is said to
be \emph{asymptotic to the identity} when outside a compact subset one has
\[
\Psi(x)=\exp_x\bigl(\zeta(x)\bigr),
\]
with $\zeta\in\mathscr{U}$ a smooth vector field.
\label{def.psiasid}\end{engdef}

\begin{engrk} \label{rk.asrig} We shall prove that the total charge is invariant
under a change of diffeomorphism at infinity that is asymptotic to the identity.
This is a restrictive assumption. Notice however that
there are many cases of interest where any $\Psi$ such that $\Psi^*h_0$
satisfy appropriate decay towards $h_0$ may be written $\Psi_0\circ A$,
where $A$ fixes $h_0$ and $\Psi_0$ is asymptotic to the identity in the
sense given above. In this case, the total charge becomes a genuine
linear form on $\mathscr{N}_0$, equivariant under the group of automorphisms
of $h_0$. We call this feature \emph{asymptotic rigidity} of the background datum.

Examples of asymptotically rigid backgrounds are the Euclidean space $\R^n$ (where
$h_0$ is the canonical flat metric)
\cite{Chrusciel-mass},\cite{Bartnik-mass}, the hyperbolic space
\cite{ChrNagy-massAAdS}, \cite{CH-massAH} or more generally a symmetric rank-$1$
space of non-compact type \cite{Herzlich-massAH}, or, in a more complicated
way, the Minkowski space at spatial infinity \cite{Chrusciel-massAMink}.
Asymptotic rigidity is however not general: for example it is well known not
to hold at null infinity in the Minkowski space when the decay conditions allow
gravitational radiation.
\end{engrk}
Let us write $h_1\defeq\Psi_1^*h$, $e_{1}\defeq h_1-h_0$
and $e_2\defeq\Psi_2^*h-h_0=\Psi^*h_1-h_0$. One computes
\begin{align*}
e_2-e_1 &=\Psi^*h_1-h_1 = \Psi^*(h_0+e_1)-(h_0+e_1)\\
   &= \Lieder_\zeta h_0 + R_1
\end{align*}
where $\Lieder$ is the Lie derivative, and
\begin{equation}
R_1\defeq(\Psi^*-\Id-\Lieder_\zeta)h_0+(\Psi^*-\Id)e_1
\label{defR1}\end{equation}
is a remainder controlled by an expression quadratic in $\zeta$, $e_1$ and their first
derivatives. From an abstract point of view, this control may be expressed in terms of
upper bounds for the first derivatives of $\exp$ and the second derivatives of $h_0$
along the geodesic $t\mapsto\exp_x\bigl(t\zeta(x)\bigr)$, \emph{i.e} in terms of upper
bounds for the Riemann tensor $\Rm_0$ of $g_0$ and for $\nablao^2h_0$ along this
geodesic. Similar control applies to the iterated covariant
derivatives of $R_1$ (see Proposition \ref{prop.R1} in the \hyperref[sec.controlR1]{Appendix}).

Therefore
\begin{equation}
\mathbb{U}(V,e_2) - \mathbb{U}(V,e_1) = \mathbb{U}(V,\Lieder_\zeta h_0)
+ R_2 \label{diffU}
\end{equation}
where $R_2\defeq\mathbb{U}(V,R_1)$ is controlled by an expression quadratic
in $\zeta$ and $e_1$ and their derivatives up to the order of $\Phi$ (as a
differential operator). These terms may be considered as second order error,
and will not contribute to the limit $k\to\infty$ in Definition \ref{defcharge}
provided the decay conditions are well chosen. This justifies the second
assumption of our main result:

\begin{engthm} \label{mainres}
Assume that
\begin{enumerate}
\item \label{hyp.mr.invphi} $\Phi(h_0)$ is \emph{invariant} under flows of vector fields,
\item \label{hyp.mr.asid} the diffeomorphisms at infinity $\Psi_1$ and $\Psi_2$ are such
   that $\Psi_1^{-1}\circ\Psi_2$ is asymptotic to the identity in the sense of Definition
   \ref{def.psiasid},
\item \label{hyp.mr.decay} the family $(S_k)$ in Definition \ref{defcharge} and the section
   $V\in\mathscr{N}_0$, are such that $R_2$ defined above satisfies
   \[
   \sup_{S_k}\abs{R_2}_{g_0}\times \Vol_{g_0}(S_k)\xrightarrow[k\to\infty]{}0.
   \]
\end{enumerate}
Then $m(h,\Psi_1,V)=m(h,\Psi_2,V)$.
\end{engthm}

Assumption \ref{hyp.mr.decay} is \emph{ad hoc}. When the curvature of $g_0$ and its
covariant derivatives are bounded, as well as $k_0$ and its derivatives, it can
be replaced by a statement easier to check: see Corollary \ref{coroR1}.

\begin{proof} From Assumption \ref{hyp.mr.decay}, one has
\[
m(h,\Psi_1,V)-m(h,\Psi_2,V)=
   \lim_{k\to\infty}\oint_{S_k}\mathbb{U}(V,\Lieder_\zeta h_0)(\nu)\ud S
\]
The theorem then follows from the \hyperref[lemma.canc]{Cancellation Lemma} below. The invariance
condition on $\Phi(h_0)$ is needed there.
\end{proof}

\begin{engrk}
Invariance of $\Phi(h_0)$ under $\Diff_0(M)$ is equivalent to the fact
that $\Phi(h_0)$ is a constant section of a trivial factor---\emph{i.e.} a
factor associated to a trivial representation of the linear group---of the bundle $F_0$.
The proof of the lemma shows in fact that the pointwise dot products
$\langle V,\Lieder_{\zeta}\Phi(h_0)\rangle_0$, when $V$ varies in $\mathscr{N}_0$
and $\zeta$ among vector fields, are obstructions to $\mathbb{U}(V,\Lieder_\zeta h_0)$
being a divergence. Therefore the invariance condition is presumably almost necessary.
\end{engrk}

\begin{Canclemma} Assume that the flows of vector fields leave $\Phi(h_0)$ invariant.
Then there exists a differential operator, equivariant under diffeomorphisms,
\[
\mathbb{V}:\Gamma(F_0\times_{M_0} H_0\times_{M_0} TM_0)\longrightarrow
   \Gamma(\Lambda^2M_0)
\]
such that for any $V \in \mathscr{N}_0=\ker \uD\Phi_0^*$ and any vector field
$\zeta$ on $M_0$ one has
\[
\mathbb{U}(V,\Lieder_\zeta h_0)=\divg_0\mathbb{V}(V,h_0,\zeta).
\]
\label{lemma.canc}\end{Canclemma}

\begin{engrk}
This lemma is a purely algebraic consequence of the definitions of
$\mathbb{U}$ and $\mathscr{N}_0$ and the $\Diff_0(M_0)$-invariance
of $\Phi_0$. The non-compactness, decay, \emph{etc.}, assumptions are
totally irrelevant here.
\end{engrk}

\begin{proof}
Let $\zeta$ be any vector field.
We apply Definition \ref{defU} with $\eta=\Lieder_{\zeta}h_0$ and $V\in\mathscr{N}_0$:
\[
\divg_0\mathbb{U}(V,\Lieder_{\zeta}h_0)
 =\bigl\langle V,\uD\Phi_0(\Lieder_{\zeta}h_0)\bigr\rangle_0
 =\bigl\langle V,\Lieder_\zeta\Phi_0\bigr\rangle_0
 =0
\]
using that $\uD\Phi_0(\Lieder_{\zeta}h_0)=\Lieder_{\zeta}\Phi_0$ because of
diffeomorphism invariance of the operator $\Phi$, and $\Lieder_{\zeta}\Phi_0=0$
from the assumption.

Assume now that $M_0$ is orientable. The Hodge star $\ast_0$ of $g_0$ conjugates
$\divg_0$ and the de Rham differential \cite[1.56 and errata]{Besse}, so that
$\zeta\mapsto\ast_0\mathbb{U}(V,\Lieder_\zeta h_0)$ is an operator:
\begin{itemize}
\item whose dependence on the triple $(h_0,V,\zeta)$ and on the orientation
is equivariant under diffeomorphisms,
\item whose values, when $h_0$ and $V\in\mathscr{N}_0$ are fixed and
$\zeta$ varies, are closed $(n-1)$-forms.
\end{itemize}
From a theorem of Wald \cite{Wald-closedforms}, there exists a form-valued operator
$\mathbb{V}^*$, equivariant under diffeomorphisms, such that for all $\zeta$,
$\ast_0\mathbb{U}(V,\Lieder_\zeta h_0)=d\mathbb{V}^*(h_0,V,\zeta)$.
The operator $\ast_0\mathbb{V}^*$ is the $\mathbb{V}$ claimed
in the theorem.

In the non-orientable case, the above reasoning works in the orientation
cover, and the formula $\mathbb{U}=\divg_0\mathbb{V}$ there projects down to $M_0$
(because through the canonical involution of the orientation cover the pulled-back
$\mathbb{U}$, $V$ and $h_0$ do not change and the Hodge star is changed into its opposite.
Following Wald's argument, $\mathbb{V}^*$ is changed into its opposite too, so $\mathbb{V}$ is not).
%
\end{proof}

\begin{engrk} 
Notice that the discussion preceding Theorem \ref{mainres} suggests that, if $\Phi$
is of order $\ell$, the decay conditions on $\Psi-\Id$ and $\Psi^*h-h_0$ should
concern $\ell$ derivatives of $\zeta$ and $h$. However it is well known that the
definition of the ADM mass for example requires only decay of the first derivatives of the
asymptotically flat metric, although it comes from
the order-$2$ scalar curvature operator (see Section \ref{sec.ex.scal.af}).

The reason is that we have studied here geometric invariance under a diffeomorphism $\Psi$
acting on $h$ and leaving $h_0$ fixed. The other point of view in the literature
is making $\Psi$ act on $h_0$ and leave $h$ fixed. These two points of view are
conjugated by a diffeomorphism acting on both $h$ and $h_0$, and are therefore equivalent.
However in the latter, $\Psi$ acts on $\mathbb{U}$, $V$, $\divg_0$, and so on,
so that the formal study is much more complicated. But with that point of view, control
on one less derivative is needed.

To give some detail, in the second point of view one needs to prove that
\[
(\Psi^*\mathbb{U})\bigl[(\Psi^*V),h-(\Psi^*h_0)\bigr]-\mathbb{U}(V,h-h_0)
\]
is a divergence up to terms at least quadratic in $\zeta$ (of Definition \eqref{def.psiasid})
and $h-h_0$. One checks that these
terms involve the derivatives of $h$ up to order $\ell-1$ only.%
\footnote{One may notice that the quadratic remainder
$Q(V,e)$ that appears in Part \ref{hyp.totch.q} of Definition \ref{hyp.totch}
involves in general the derivatives of $e$ up to order $\ell$.
This can be dealt with if $\Phi$ is quasilinear of order at least $2$, see Remark
\ref{rk.intQ}.}
\label{rk.limR2}\end{engrk}

\section{Examples} \label{sec.ex}

We recover here known examples, to the invariance of which Theorem \ref{mainres} gives
a formal proof.

\subsection{The scalar curvature operator} \label{sec.ex.scal}
First we take for $\Phi$ the scalar curvature operator. Here $E$, $E_0$ are the null
bundles over respectively $M$ and $M_0$, and $F$, $F_0$ are the trivial line bundles.
Let us write $g_0$ for the reference metric. We note $\nablao$ its Levi-Civita covariant
derivative (without the index $0$ to trim notations) and $\Ric_0$ its Ricci tensor.
Let $g=g_0+e$ be another metric on $M_0$ (at least outside a compact subset).
We have
\[
\Scal^{g_0+e}=\Scal^{g_0}+\uD\Scal_0(e)
  +Q(1,e),
\]
where the linearization of the scalar curvature at $g_0$ is, \emph{cf.} Besse
\cite{Besse}:
\[
\uD\Scal_0(e)= \divg_0\bigl(\divg_0e-d(\tr e)\bigr) -
      \langle\Ric_{0},e\rangle_0.
\]
Here traces are taken with respect to $g_0$, and
$\divg_0$ is the divergence operator of $g_0$: for a multi-index $J$
and a tensor $T_{iJ}$, $\divg_0T_J=\nablao^iT_{iJ}$.

The formal adjoint of $\uD\Scal_0$ is
\[
\uD\Scal_0^*(V)=\nablao^2V + \Delta_0V g_{0} - V \Ric_{0}
\]
where $\Delta_0=-\nablao^i\nablao_i$ is the geometric Laplacian.
The equation $\uD\Scal_0^*V=0$ is equivalent to the
fact that the metrics $g_0\pm V^2d t^2$ on $M_0\times \R$ are Einstein,
\emph{cf.} \cite{CH-massAH}.

The charge boundary integrand computed with Definition \ref{defU} is:
\begin{equation}
\mathbb{U}(V,e)= V\bigl(\divg_0e-d(\tr e)\bigr)
   - \imath_{\nablao V}e + (\tr e)d V
\label{Uscal}\end{equation}
where $\imath_X$ denotes the contraction of a vector $X$ with a (covariant) tensor.
Up to second-order terms this is the formula of \cite{CH-massAH}, see also
\cite{Herzlich-massAH}.

One checks also that, if the operator norm of $g_0^{-1}e$ is not greater
than $1/2$ (so that, with the help of the Neumann series,
$\abs{g^{-1}-g_0^{-1}}_0\leq2\abs{e}_0$ and
$\abs{\nablao g^{-1}}_0 \leq 4\abs{\nablao e}_0$), the quadratic remainder
of Equation \eqref{eqU} is bounded:
\begin{equation}
Q(1,e)\leq C\bigl(\abs{\nablao e}_0^2+\abs{e}_0\abs{\smash[t]{\nablao^2} e}_0\bigr)
\label{quadscal}\end{equation}
where $C$ is a dimensional constant (independent of $g_0$ in particular).
\begin{engrk}
This suggests that the appropriate decay conditions for the definitions
of masses and center of mass---to be presented below---should concern
two derivatives of $g-g_0$. But it is well known that only control
on the first derivatives are needed. This comes from the fact that,
using an integration by part whenever a second derivative of $e$ occurs,
one may write
\[
Q(V,e)=VQ(1,e)=Q_1(V,e)+Q'_1(\nablao V,e)+\divg_0 Q_2(V,e),
\]
where $Q_1$, $Q'_1$ and $Q_2$ are linear in their first argument, and
at least quadratic in $e$ and its first derivatives. Thus in Definition \ref{hyp.totch},
integrability of $Q(V,e)$ may be replaced by decay of $e$ and $\nablao e$ such that
$Q_1$ and $Q'_1$ are integrable and integrals of $Q_2$ over large spheres vanish in the
limit.
\label{rk.intQ}\end{engrk}

\subsubsection{The asymptotically flat case} \label{sec.ex.scal.af}
First we take $M_0=\R^n$ and $g_0$ the canonical flat metric.

This background space is asymptotically rigid in the sense of Remark \ref{rk.asrig}, see
\cite{Chrusciel-mass} and \cite{Bartnik-mass}. Namely,
any two systems of coordinates at infinity on $(M,g)$ of class
$C^2$, such that in both, the coefficients of $g$ satisfy:
\begin{equation}
e_{ij}\defeq g_{ij}-\delta_{ij}=O(r^{-\tau}),\quad \partial_kg_{ij}=O(r^{-\tau-1})
\quad \text{with } \tau>0,
\label{aflat}\end{equation}
(with $r=(\sum_ix^{i\,2})^{1/2}$),
differ in a diffeomorphism at infinity $\Psi\circ A$, where $A$ is a Poincaré
transform $A:x\mapsto Rx+T$, $R\in O(n)$, $T\in \R^n$, and $\Psi$ is
asymptotic to the identity in the sense of Definition \ref{def.psiasid},
with $\zeta$ there satisfying
\begin{equation} \label{zetaaf}
\zeta^i=O(r^{1-\tau}),\qquad\partial_j\zeta^i=O(r^{-\tau}),
   \qquad\partial_{j}\partial_{k}\zeta^i=O(r^{-\tau-1})
\end{equation}
(here of course $\exp_x\zeta(x)$ has coordinates $x^i+\zeta^i(x)$).

The space $\mathscr{N}_0$ consists of affine functions, which grow like $r$
at infinity. Part \ref{hyp.totch.lc} of Definition \ref{hyp.totch}
is satisfied if $r\,\Scal^g$ is integrable.
Part \ref{hyp.totch.q} of that definition is satisfied if \eqref{aflat}
and $\partial_k\partial_lg_{ij}=O(r^{-\tau-2})$ hold with $\tau>\frac{n-1}{2}$,
since, because of \eqref{quadscal}, the quadratic term $Q(V,e)$ is then
$O(r^{-2\tau-1})$, with $-2\tau-1<-n$.
As explained in Remark \ref{rk.intQ}, the control of the second derivatives
of $g$ is in fact superfluous.

Because of the asymptotic rigidity stated above, if we can apply Theorem \ref{mainres}
to a vector field $\zeta$ satisfying \eqref{zetaaf} for appropriate $\tau$, we will
obtain an actual invariance of the total-charge linear form on $\mathscr{N}_0$.
We could apply Corollary \ref{coroR1}, but it is not optimal for the usual asymptotic
flatness assumptions \eqref{aflat}. Let us instead compute in coordinates for example
the last term of $R_1$ in \eqref{defR1}:
\begin{multline}
\zeta^c(x)\int_0^1\partial_ce_{ij}(x+t\zeta)d t + e_{aj}(x+\zeta)\partial_i\zeta^a(x)
   +e_{ib}(x+\zeta)\partial_j\zeta^b(x)\\
   +e_{ab}(x+\zeta)\partial_i\zeta^a(x)\partial_j\zeta^b(x),
\label{R1right.af}\end{multline}
which is $O(r^{-2\tau})$ if \eqref{aflat} and \eqref{zetaaf} hold, and whose derivatives
are $O(r^{-1-2\tau})$ if moreover $\partial_k\partial_lg_{ij}=O(r^{-\tau-2})$.
The same estimates apply to the first term of the right-hand side of \eqref{defR1},
which is here $\sum_k\partial_i\zeta^k\partial_j\zeta^k$.
So the term $R_2$ of Equation \eqref{diffU} is $O(r^{-2\tau})$ under these
assumptions. If $2\tau>n-1$, its integral over large coordinates spheres vanishes in
the limit. This is what is needed to apply Theorem \ref{mainres}.
(Again, in view of Remark \ref{rk.limR2}, the control of the second derivatives
of $g$ is in fact superfluous.)

Therefore, under these decay assumptions we recover the well-known ADM mass:
\[
\mathbb{U}_{j}(1,e)=\sum_{i}\partial_ie_{ij}-\partial_{j}e_{ii}
\]
and center of mass:
\[
\mathbb{U}_{j}(x^a,e)=\sum_{i}\bigl[x^a(\partial_ie_{ij}-\partial_{j}e_{ii})
   -e_{ai}+e_{ii}\delta_{aj}\bigr].
\]

Of course the subspace $\mathscr{N}'_0$ of constant functions is invariant under
the isometries of $\R^n$. Since constant functions grow slower at infinity than
general affine functions, this allows to relax the asymptotic decay conditions
to define the ADM mass: as is well known, its definition only requires
$\Scal^g$ to be integrable and $\tau>\frac{n-2}{2}$ in \eqref{aflat}.

\subsubsection{The asymptotically hyperbolic case}\label{sec.ex.scal.ah}

Here we set $M_0=\Hyp^n$ endowed with the hyperbolic metric
$g_0=d r^2+(\sinh^2r)\breve{g}$, where $\breve{g}$ is
the canonical metric of the unit sphere $\Sph^{n-1}$ and $r$ the distance to a fixed point.
We refer to \cite{ChrNagy-massAAdS}
and \cite{CH-massAH} for a detailed treatment, which includes more general, similar cases
(notice that their coordinate $r$ is the hyperbolic sine of ours).

Asymptotic rigidity holds here \cite[Theorem 3.3]{ChrNagy-massAAdS},
\cite[proof of Theorem 2.3]{CH-massAH}, as mentioned in Remark \ref{rk.asrig}: if
a metric $g$, pulled-back on $\Hyp^n$ \emph{via} two diffeomorphisms $\Psi_{q}$,
$1\leq q\leq2$, satisfies
\begin{equation}
\abs{\Psi_q^*g-g_0}_0=O(e^{-\tau r})\quad\text{and}\quad\abs{\nablao (\Psi_q^*g)}_0=O(e^{-\tau r})
\label{ahyp}\end{equation}
for some $\tau>1$, then $\Psi_1^{-1}\circ\Psi_2=\Psi\circ A$, where $A$ is an isometry of
$g_0$, and $\Psi$ is asymptotic to the identity in the sense of Definition \ref{def.psiasid},
with
\begin{equation}\label{zetaah}
\abs{\zeta}_0, \abs{\nablao\zeta}_0, |\nablao^2\zeta|_0=O(e^{-\tau r}).
\end{equation}
We set $e\defeq\Psi^*_1g-g_0$.

The space $\mathscr{N}_0$ is generated by the coordinates of the canonical isometric
embedding
\[
(V_{(0)},V_{(1)},\dotsc,V_{(n)}):\Hyp^n\to\R^{1,n}
\]
into the Minkowski
space. In spherical coordinates one has $V_{(0)}=\cosh r$ and
$V_{(i)}=(\sinh r)\xi^i$, $i\geq1$, where $(\xi^1,\dotsc,\xi^n):\Sph^{n-1}\to\R^n$ is
the canonical embedding. These functions and their $g_0$-gradient are $O(e^r)$.
So Part \ref{hyp.totch.lc} of Definition \ref{hyp.totch}
is satisfied if $e^r\,\Scal^g$ is integrable.
Part \ref{hyp.totch.q} of that definition is satisfied if \eqref{ahyp} and
$|\nablao^2g|_0=O(e^{-\tau r})$ hold with $\tau>\frac{n}{2}$, since then
\[
Q(V,e)\ud\vol_0 = Q(V,e)(\sinh^{n-1}r)\ud r\ud\vol_{\breve g}
=  O(e^{(n-2\tau)r})\ud r\ud\vol_{\breve g}
\]
because of \eqref{quadscal}, and $n-2\tau<0$. (As explained in Remark \ref{rk.intQ}
the control of the second derivatives of $g$ is in fact superfluous.)

As for changes of asymptotically hyperbolic coordinates, because of asymptotic rigidity,
Theorem \ref{mainres} applied to $\zeta$'s satisfying \eqref{zetaah} (with appropriate $\tau$)
will yield an actual invariance of the total-charge form on $\mathscr{N}_0$. We need to check
Assumption \ref{hyp.mr.decay}: we use Corollary \ref{coroR1}. We take for $S_r$ the
distance spheres. Their volume is $\abs{S_r}=\sinh^{n-1}r$. We have $\abs{V}+\abs{\nabla V}
=O(e^r)$, and $\abs{\mathbb{U}(V,e)}\leq U (\abs{V}+\abs{\nabla V})(\abs{e}+\abs{\nabla e})$
for some constant $U$. Hence the three conditions of Corollary \ref{coroR1} are satisfied
when \eqref{ahyp}, \eqref{zetaah} and $|\nablao^2g|_0=O(e^{-\tau r})$ hold with
$\tau>\frac{n}{2}$. (Again, in view of Remark \ref{rk.limR2}, the
control of the second derivatives of $g$ is in fact superfluous.)

Thus, under the condition $\tau>\frac{n}{2}$ in \eqref{ahyp}, we recover here the
Chru\'sciel-Herzlich momentum in the asymptotically hyperbolic
setting, given by the limit of the integrals of \eqref{Uscal} over large $r$-spheres,
when $V$ is one of the $V_{(\mu)}$ above. These $V_{(\mu)}$ are reshuffled by a Lorentz
matrix through the action of the isometries of $g_0$. Therefore so are the
associated total charges.

\subsection{The operator of constraints of general relativity}
   \label{sec.ex.constr}
   
Here we take for $\Phi$
the constraints on Cauchy initial data in General Relativity. The geometric
data are a couple $(g,k)$ of a first and a second fundamental form of a
hypersurface in an ambient space-time, and
\begin{align*}
\Phi:\G(\mathscr{M} \times_M S_2M)&\longrightarrow\G(\R \oplus T^*M)\\
   (g,k)&\longmapsto\left(\begin{array}{l}\Scal^g+(\tr_gk)^2-\abs{k}^2_g\\
              2\bigl(\divg_gk-d(\tr_gk)\bigr) \end{array}\right)
              \eqdef\left(\begin{array}{c}\Phi^H(g,k)\\
                             \Phi^M(g,k)\end{array}\right).
\end{align*}
The invariance condition for $\Phi_0$ imposes $\Phi^M_0$ to vanish
and $\Phi^H_0$ to be a constant $2\Lambda$: these
are the constraint equations with cosmological constant in vacuum for
$(g_0,k_0)$. The test-section $V$ is a couple $(f,\alpha)$
of a function and a $1$-form over $M_0$. The linearized operators at $(g_0,k_0)$,
for a variation $e=(\dot g,\dot k)$, are
\begin{align*}
\uD\Phi^H_0(e)&=\divg\divg\dot g + \Delta\tr\dot g
   -\bigr\langle\Ric_0-2k_0\circ k_0+2(\tr k_0)k_0,\dot g\bigl\rangle\\
   &\quad-2\langle k_0,\dot k\rangle + 2\tr k_0 \tr \dot k\\
\uD\Phi^M_0(e)&=\imath_{\nablao\tr\dot g}k_0-k_0^{ij}\nablao\dot g_{ij}
   -2\bigl(\divg(\dot g\circ k_0)-d\langle k_0,\dot g\rangle\bigr)\\
   &\quad +2\bigl(\divg\dot k-d(\tr \dot k)\bigr)
\end{align*}
where $\Ric_0$ is the Ricci tensor of $g_0$; traces, divergences,
index lowering and raising (implicit when needed), \emph{etc.} are taken with respect
to $g_0$; and $\circ$ is the composition of $2$-tensors (using $g_0$):
$(A\circ B)_{ij}\defeq g_0^{kl}A_{ik}B_{lj}$.
From there one computes (beware of the sign convention: here it is Riemannian)
\begin{align*}
\bigl\langle V,\uD\Phi_0(e)\bigr\rangle&=f\uD\Phi_0^H(e)
      +\bigl\langle\alpha,\uD\Phi_0^M(e)\bigr\rangle\\
   &=\divg\mathbb{U}(V,e)
      +\bigl\langle(\partial_g\Phi)_0^*(V),\dot g\bigr\rangle
      +\bigl\langle(\partial_k\Phi)_0^*(V),\dot k\bigr\rangle\\
\intertext{with}
(\partial_g\Phi)_0^*(f,\alpha)&=\nablao^2f+(\Delta f)g_0
      -f\bigl(\Ric_0-2k_0\circ k_0+2(\tr k_0)k_0\bigr)\\
   &\quad+\Lieder_{\alpha}k_0 - (\divg\alpha)k_0-
      \bigl(\langle\nablao\alpha,k_0\rangle
      +\langle\alpha,\divg k_0\rangle\bigr)g_0\\
(\partial_k\Phi)_0^*(f,\alpha)&=-2f(k_0-\tr k_0 g_0)
      -\Lieder_\alpha g_0+2(\divg\alpha)g_0
\end{align*}
and
\begin{align*}
\mathbb{U}(V,e)&=f\bigl(\divg\dot g-d(\tr\dot g)\bigr)
        -\imath_{\nablao f}\dot g + (\tr\dot g) d f\\
     &\quad+2 \bigl(\imath_{\alpha}\dot k-(\tr\dot k)\alpha\bigr)
        +(\tr\dot g)\imath_{\alpha}k_0+\langle k_0,\dot g\rangle\alpha
        -2\imath_{\alpha}(\dot g \circ k_0).
\end{align*}
Up to terms quadratic in $\dot g$ and $\dot k$,
this is the general formula of Chru\'sciel, Jezierski and \L\k{e}ski
\cite[§2]{CJL-Bondimass} (notice that their $Y$ is $-g_0^{-1}\alpha$ here, and their
$P$ is not in fully covariant form, whereas $k$ here is). See also \cite{Maerten-PosMomt}
(beware of the sign conventions for the momentum constraints and the
divergence operator).

One sees also that the quadratic error $Q(V,e)$ of \eqref{eqU}, that appears in Definition
\ref{hyp.totch}, is bounded, up to a multiplicative constant independent of $g_0$ and $k_0$,
by
\begin{multline}
\abs{f}\Bigl(\abs{\dot g}^2\abs{k_0}^2 +\abs{\nablao\dot g}^2
   +\abs{\dot g}\abs{\nablao^2\dot g}+|\dot k|^2\Bigr)\\
+\abs{\alpha}\Bigl(\abs{\dot g}^2(\abs{k_0}^2+\abs{\nablao k_0})
   +\abs{\nablao\dot g}^2+|\dot k|^2+\abs{\dot g}|\nablao \dot k|\Bigr)
\label{quadconstr}\end{multline}
when the operator norm of $g_0^{-1}\dot g$ is no more than $1/2$.
(As in Remark \ref{rk.intQ}, the assumption of integrability of $Q(V,e)$
may be replaced by decay conditions on $e$ that involve neither $\nablao^2\dot g$
nor $\nablao\dot k$.)

It is known \cite{Moncrief-sym} that the kernel $\mathscr{N}_0$ of the
operator $\uD\Phi_0^*$ contains exactly the couples $(f,\alpha)$ such
that $(f,g_0^{-1}\alpha)$ is the normal-tangential decomposition
of the restriction along $M_0$ of a Killing vector field of the
(Lorentzian) Choquet-Bruhat development of $(M_0,g_0,k_0)$.

If one wants an asymptotic rigidity result as mentioned in Remark \ref{rk.asrig}
for a couple $(g_0,k_0)$, one may consider $g_0$ as in the examples of
Section \ref{sec.ex.scal} (flat or hyperbolic), and $k_0=\lambda_0g_0$. The
invariance condition on $\Phi(g_0,k_0)$ imposes $\lambda_0$ to be
constant, so that isometries of $g_0$ also fix $k_0$. The boundary integrand then
reduces to
\[
\mathbb{U}(V,e)=\bigl(f\divg_0-\imath_{\nablao f}\bigr)\bigl(\dot g-(\tr\dot g)g_0\bigr)
   +2\imath_\alpha\bigl(\dot k -\lambda_0\dot g - \tr(\dot k-\lambda_0\dot g)g_0\bigr).
\]
This is the Chru\'sciel-Jezierski-\L\k{e}ski expression for the Bondi mass
\cite[§3]{CJL-Bondimass}, under asymptotic decay conditions that do however not
allow gravitational radiation. When $\lambda_0=0$, $\mathbb{U}$ decomposes into
a $(f,\dot g)$-part equal to the $\mathbb{U}$ of Section \ref{sec.ex.scal}, and a
$(\alpha,\dot k)$-part: $2\imath_{g_0^{-1}\alpha}(\dot k-\tr_0\dot kg_0)$,
that give respectively the usual ADM and Abbott-Deser momenta.

\appendix

\section{Explicit control of the remainder %
\texorpdfstring{$R_1$}{R1}}
   \label{sec.controlR1}

Recall that we are given the data $h_0=(g_0,k_0)$ of a Riemannian metric $g_0$ and a natural
tensor $k_0$ over $M_0$. We assume that $g_0$ is complete. We consider a diffeomorphism
$\Psi:x\mapsto\exp_x\zeta(x)$ (where $\exp$ is the $g_0$ exponential map), and search for
estimates of $\nabla^\ell(\Psi^*k-k)$ and $\nabla^\ell(\Psi^*k-k-\Lieder_\zeta k)$, when
$k$ is a tensor field.

We first introduce some notations. Let
$\abs{k}_\ell(x)\defeq\abs{k(x)}_0+\cdots+\abs{\nablao^\ell k(x)}_0$, and
$\norm{k}_\ell(x)$ be the supremum of $\abs{k}_\ell$ along the geodesic $t\mapsto\exp_xt\zeta(x)$,
$t\in[0,1]$.

\begin{prop}
\begin{enumerate}
\item \label{propR1.diff}
  There exists a universal constant $\veps>0$ such that, if $\zeta$ and the sectional curvatures
  $\kappa_0$ of $g_0$ satisfy $\kappa_0\abs{\zeta}_0^2\leq\veps$ and $\abs{\nabla\zeta}\leq\veps$
  on $M_0$, then $\Psi\defeq\exp\circ\zeta$ is a diffeomorphism such that
  \[
  \frac{1}{4}g_0\leq\Psi^*g_0\leq4g_0.
  \]
\item \label{propR1.contr} Assume that Point \ref{propR1.diff} holds.
   Assume that the Riemann tensor $\Rm_0$ of $g_0$, $\zeta$ and their covariant derivatives
   up to order $\ell$ are bounded. Then there exists a constant
   $C$ such that the term $R_1$ in Equation \eqref{defR1} satisfies
   \[
   \abs{\nablao^\ell R_1}_0\leq
      C\abs{\zeta}_{\ell+1}\bigl(\abs{\zeta}_{\ell+1}\norm{k_0}_{\ell+2}
         +\norm{e_1}_{\ell+1}\bigr).
   \]
\end{enumerate}
\label{prop.R1}\end{prop}

\begin{proof}[Proof of Proposition \ref{prop.R1}]
Since norms, covariant derivatives and curvatures always refer to $g_0$ in this proof, we shall
omit the indices $0$. The notation $\abs{\cdot}$ will here denote the
pointwise $g_0$-norm of its argument, and $\norm{\cdot}$ its supremum along the geodesic
$t\mapsto\exp_xt\zeta(x)$. The letter $C$ will refer to a positive constant that may
change from line to line.\\

{\it Proof of Part \ref{propR1.diff}}. Let us denote $\Psi_t(x)\defeq\exp_xt\zeta(x)$.
We fix $x\in M_0$, an orthonormal frame $E_i$
of $T_xM_0$, and a vector $X\in T_xM_0$, and write $\gamma(t)\defeq\Psi_t(x)$ for $t\in[0,1]$.
Let $T(t)\defeq\DP{\gamma}{t}$. For short we will write $\nabla_T$ instead of
$\gamma^*(\nabla)_{\partial_t}$ or $\frac{\nabla}{\ud t}$.
Let $E_i(t)$ be the parallel transport of $E_i$ along $\gamma$. When $k$ is a section of
a tensor bundle $K$, we denote by $k(t)$ the value of $k(\gamma(t))$ in the trivialization
of $\gamma^*K$ given by the frame $E_i(t)$. Note that $k'(t)=(\nabla_Tk)(t)$.

Let $X(t)\defeq T_x\Psi_tX$. It is the Jacobi field along $\gamma$ with initial conditions
$X(0)=X$, $X'(0)=\nablao_X\zeta(x)$. Classical comparison techniques insure that, provided
with sufficiently small upper bounds for $\abs{T(t)}^2\kappa(\gamma(t))$
(where $\kappa$ is the pointwise supremum of the sectional curvatures)
and $\abs{\nabla\zeta(x)}$, one has $1/2\abs{X(0)}\leq\abs{X(t)}\leq 2\abs{X(0)}$ along
$\gamma$ (which implies that $\Psi$ is a diffeomorphism). Such upper bounds hold if
$\kappa\abs{\zeta}^2$ and $\abs{\nabla\zeta}$ are bounded by a sufficiently small $\veps>0$
everywhere on $M_0$. This gives Part \ref{propR1.diff} of the proposition. From now on we assume
that it holds.\\

{\it Proof of Part \ref{propR1.contr} when $\ell=0$}. Let $P(t)$ be the matrix of
$T_x\Psi_t$ with respect to the frames $E_i(0)$ and $E_i(t)$. It is bounded in a
fixed neighborhood of the identity matrix.
From the Jacobi equation $\abs{P''(t)}\leq C \norm{\Rm}(x)\abs{\zeta}^2(x)$ and therefore,
for $t\leq1$, because of the initial conditions:
\[
\abs{P'(t)}\leq C\bigl(\norm{\Rm}(x)\abs{\zeta}^2(x)+\abs{\nabla\zeta}(x)\bigr).
\]
For a tensor field $k$, $\Psi_t^*k(x)$ has value $P(t)^{-1}\cdot k(t)$ in the frame $E_i(0)$ (the
action of $GL_n$ defining the tensor bundle in which $k$ lives is understood). Thus
$1/C \abs{k}\leq\abs{\Psi^*k}\leq C\abs{k}$ for some positive constant $C$;
\begin{align}
\abs{\Psi^*k(x)-k(x)}&\leq
   \int_0^1\left|\frac{\ud}{\ud t}\bigl(P(t)^{-1}\cdot k(t)\bigr)\right|\ud t
   \leq C\Bigl[\bigl(\norm{\Rm}\abs{\zeta}^2+\abs{\nabla\zeta}\bigr)\norm{k}
      +\abs{\zeta}\norm{\nabla k}\Bigr](x) \notag\\
   &\leq C\abs{\zeta}_1(x)\norm{k}_1(x);\label{R1right}
\end{align}
when $\Rm$ and $\zeta$ are bounded; and with the help of a Taylor formula
\begin{align}
\abs{\Psi^*k(x)-k(x)-\Lieder_\zeta k(x)}&\leq
   \int_0^1(1-t)\left|\frac{\ud^2}{\ud t^2} \bigl(P(t)^{-1}\cdot k(t)\bigr)\right|\ud t\notag\\
   &\leq C\Bigl[\abs{\zeta}^2\norm{\nabla^2k}
      +\abs{\zeta}\bigl(\norm{\Rm}\abs{\zeta}^2+\abs{\nabla\zeta}\bigr)\norm{\nabla k}
      +\norm{\Rm}\abs{\zeta}^2\norm{k}\Bigr](x)\notag\\
   &\leq C\abs{\zeta}_1^2(x)\norm{k}_2(x).\label{R1left}
\end{align}
Plugging Equations \eqref{R1right} and \eqref{R1left} into \eqref{defR1} with \resply\ $k=e_1$
and $k=k_0$ gives Part \ref{propR1.contr} of
the proposition in case $\ell=0$.\\

{\it Proof of Part \ref{propR1.contr} when $\ell=1$}. The proof is an
induction on $\ell$. We present the first step in detail. We parallel-transport the previous
construction along a geodesic through $x$: let us set
\begin{itemize}
\itbul $X,Y\in T_xM_0$,
\itbul $\gamma_1(s)\defeq\exp_x(sX)$,
\itbul $X(s,0),Y(s,0)$ the parallel transports of $X$ and $Y$ respectively along $\gamma_1$,
\itbul $\gamma(s,t)\defeq\Psi_t\bigl(\gamma_1(s)\bigr)$,
\itbul $X(s,t)\defeq T_{\gamma_1(s)}\Psi_tX(s,0)=\DP{\gamma}{s}$ the Jacobi field along
   $t\mapsto\gamma(s,t)$ with initial data $X(s,0)$ and
   $\nabla_{X(s,0)}\zeta\bigl(\gamma_1(s)\bigr)$, and $Y(s,t)\defeq T_{\gamma_1(s)}\Psi_tY(s,0)$
   similarly,
\itbul $E_i(t)$ the parallel transport of an orthonormal frame at $x$ along
   $t\mapsto\exp_xt\zeta(x)$.
\end{itemize}
We write $T(t,s)\defeq\DP{\gamma}{t}(s,t)$, and $\nabla_T=\frac{\nabla}{\partial t}$,
$\nabla_X=\frac{\nabla}{\partial s}$ for short. The covariant derivative of the Jacobi
equation satisfied by $Y$ in the direction of $X$ may be written
\begin{multline} \label{ODEnabla}
\nabla_T\nabla_T(\nabla_XY)-\Rm_{T,\nabla_XY}T=(\nabla_X\Rm)_{T,Y}T+\Rm_{\nabla_TX,Y}T
      +\Rm_{T,Y}\nabla_TX\\
   +\Rm_{T,X}\nabla_TY-\nabla_T\bigl(\Rm_{T,X}Y).
\end{multline}
This is a second order ODE for $\nabla_XY$, whose homogeneous part is the Jacobi equation
(a classical fact in ODE theory); the initial conditions are
\begin{align*}
\nabla_XY(s,0)&=0,\\
\frac{\nabla}{\partial t}\nabla_XY(s,0)&=\Rm_{T,X}Y(s,0)+
   \frac{\nabla}{\partial s}\frac{\nabla}{\partial t}Y(s,0)=
   \Rm_{T,Y}X(s,0)+(\nabla^2\zeta)_{X,Y}(s,0).
\end{align*}
If $\Rm$, $\nabla\Rm$ and $\zeta$ are bounded, the source term in the ODE \eqref{ODEnabla}
is bounded by
\[
C\abs{\zeta}_1(x)\abs{X(0,0)}\abs{Y(0,0)}
\]
because of the part $\ell=0$ of the proof
applied to $X(s,t)$ and $Y(s,t)$. Moreover
\[
\abs{\nabla_T\nabla_XY}_{t=0}\leq
   C\bigl(\abs{\zeta}+\abs{\nabla^2\zeta}\bigr)\abs{X(0,0)}\abs{Y(0,0)}.
\]
The method of variation of parameters then shows that
\[
\abs{\nabla_XY(s,t)}\leq C\abs{\zeta}_2(x),\quad
  \abs{\nabla_T\nabla_XY(s,t)}\leq C\abs{\zeta}_2(x)
\]

Let $\Gamma(t)$ be the matrix of the linear map $X(0,0)\otimes Y(0,0)\mapsto \nabla_XY(0,t)$
with respect to the frames $E_i(0)$ and $E_i(t)$. Then $\Psi_t^*\nabla-\nabla$ has matrix
$P(t)^{-1}\Gamma(t)$ in the frame $E_i(0)$ ($P(t)$ being again the matrix of $T_x\Psi_t$). From
the previous paragraph, the part $\ell=0$ of the proof and the ODE \eqref{ODEnabla} for
$\Gamma''(t)$, one has
\[
\abs{P(t)^{-1}\Gamma(t)}\leq C\abs{\zeta}_2,\quad
   \abs{\frac{\ud}{\ud t}P(t)^{-1}\Gamma(t)}\leq C\abs{\zeta}_2,\quad
   \abs{\frac{\ud^2}{\ud t^2}P(t)^{-1}\Gamma(t)}\leq C\abs{\zeta}_2^2.
\]
Therefore as in the case $\ell=0$,
\begin{equation*}
\abs{\Psi^*\nabla-\nabla}\leq C\abs{\zeta}_2,\quad
   \abs{\Psi^*\nabla-\nabla-\Lieder_\zeta(\nabla)}\leq C\abs{\zeta}_2^2.
\end{equation*}
Thus for any tensor field $k$:
\begin{align}
\left|\nabla\bigl(\Psi^*k-k\bigr)\right|&=\left|\bigl(\Psi^*-\Id\bigr)\bigl(\nabla k\bigr)
      -\bigl(\Psi^*\nabla-\nabla\bigr)\Psi^*k\right|
      \leq C \abs{\zeta}_2\norm{k}_2 \label{DR1right}\\
\left|\nabla\bigl(\Psi^*k-k-\Lieder_\zeta k\bigr)\right|&=
      \left|\bigl(\Psi^*-\Id-\Lieder_\zeta\bigr)\nabla k
      -\bigl(\Psi^*\nabla-\nabla\bigr)\bigl(\Psi^*k-k\bigr)\right.\notag\\
      &\qquad\left.-\bigl(\Psi^*\nabla-\nabla-\Lieder_\zeta(\nabla)\bigr)k\right|\notag\\
      &\leq C\abs{\zeta}_2^2\norm{k}_2 \label{DR1left}
\end{align}
(where the result for $\ell=0$ has been used as well). The case $\ell=1$ of the proposition
follows from \eqref{DR1right} with $k=e_1$ and from \eqref{DR1left} with $k=k_0$ in the
$\nabla$-derivative of Equation \eqref{defR1}.\\

The induction on $\ell$ uses a similar construction and a similar ODE for iterated derivatives
$\nabla_{X_1}\cdots\nabla_{X_l}Y$. They result in estimates
\[
\abs{\nabla^{\ell-1}(\Psi^*\nabla-\nabla)}\leq C\abs{\zeta}_{\ell+1},\quad
 \abs{\nabla^{\ell-1}(\Psi^*\nabla-\nabla-\Lieder_\zeta\nabla)}\leq C\abs{\zeta}_{\ell+1}^2
\]
when $\Rm$, $\zeta$ and their derivatives up to order
$\ell$ are bounded (which allows to easily control the otherwise complicated source term
and initial conditions of the ODE). The claimed control on $\nabla^\ell R_1$ is deduced as above.
The induction is straightforward but lengthy, so we do not give the details.
\end{proof}

We use Proposition \ref{prop.R1} to replace Assumption \ref{hyp.mr.decay} of Theorem \ref{mainres}
with a less \emph{ad hoc} statement. Let $\ell$ be the order of the differential operator $\Phi$.
Recalling that $\mathbb{U}(V,\eta)$ is a linear differential operator of order $\ell-1$ in
$V$ and $\eta$, there exists a nonnegative continuous function $U$ such that
\begin{equation} \label{defnormU}
\forall x\in M_0, \abs{\mathbb{U}(V,\eta)}\leq U(x)\abs{V}_{l-1}(x)\abs{\eta}_{l-1}(x).
\end{equation}
Noticing that the operator $\mathbb{U}$ is moreover a differential operator of order at most
$\ell$ in the reference data $(g_0,k_0)$, equivariant under diffeomorphisms, the function $U$
in Equation \eqref{defnormU} is bounded when the curvature of $g_0$, its derivatives up to order
$\ell-2$ and the derivatives of $k_0$ up to order $\ell$ are bounded. Applying \eqref{defnormU}
to $\eta=\Lieder_\zeta h_0$, 
we then obtain the following direct consequence of Proposition \ref{prop.R1}:

\begin{engcor}\label{coroR1}
Assume that Assumption \ref{hyp.mr.asid} of Theorem \ref{mainres} holds.
Assume moreover that the Riemann tensor of $g_0$ and its derivatives up to order $\ell-1$, and
$k_0$ and its derivatives up to order $\ell+1$ are bounded. Then Assumption \ref{hyp.mr.decay}
of Theorem \ref{mainres} is implied by the following conditions on the family of hypersurfaces
$S_k$, $\zeta$ (of Definition \ref{def.psiasid}) and $e_1$:
\begin{gather*}
\abs{\zeta}_{\ell-1}\text{ is bounded},\\
\Vol(S_k)\sup_{S_k}(\abs{V}_{\ell-1}\abs{\zeta}_{\ell}^2)\xrightarrow[k\to\infty]{}0,\\
\Vol(S_k)\sup_{S_k}(\abs{V}_{\ell-1}\norm{e_1}_{\ell}^2)\xrightarrow[k\to\infty]{}0.
\end{gather*}
\end{engcor}

\bibliographystyle{amsplain}
\bibliography{bibccanc}

\end{document}